\documentclass[a4paper,12pt]{article}
\usepackage[top=1in,bottom=1in,left=1in,right=1in]{geometry}

\usepackage{amssymb}
\usepackage{amsthm}
\usepackage{amsmath}
\usepackage{color}
\usepackage{graphicx}
\usepackage{hyperref}
\usepackage{cleveref}
\usepackage{natbib}
\usepackage{tikz}
\usepackage{setspace}
\usepackage{subfiles}
\usepackage{comment}

\DeclareMathOperator*{\argmax}{arg\,max}
\DeclareMathOperator*{\argmin}{arg\,min}

\theoremstyle{definition}

\newtheorem{lem}{Lemma}

\newtheorem{them}{Theorem}

\newtheorem{coro}{Corollary}

\newtheorem{prop}{Proposition}

\newtheorem{ass}{Assumption}

\title{Characterizing the Feasible Payoff Set of OLG Repeated Games\thanks{We thank the two anonymous referees and the coeditor for their helpful comments and discussions. We are grateful to the seminar and conference participants of the 2023 Asian Meeting of the Econometric Society (Beijing), the 2023 Africa Meeting of the Econometric Society (Nairobi), and the CIRJE Microeconomics Workshop at the University of Tokyo. We also thank Michihiro Kandori and Akihiko Matsui for their helpful suggestions. Chihiro Morooka was supported by the Japan Society for the Promotion of Science (JSPS) KAKENHI Grant Number JP22K13360. All remaining errors are ours.}}
\author{Daehyun Kim\thanks{Division of Humanities and Social Sciences, POSTECH. Email: \href{mailto:dkim85@outlook.com}{dkim85@outlook.com}}\and Chihiro Morooka\thanks{School of Science and Engineering, Tokyo Denki University. Email: \href{mailto:c-morooka@mail.dendai.ac.jp}{c-morooka@mail.dendai.ac.jp}}}
\date{December 5, 2024}

\begin{document}
\maketitle
\begin{abstract}
We study the set of (stationary) feasible payoffs of overlapping generation repeated games that can be achieved by action sequences in which every generation of players plays the same sequence of action profiles. First, we completely characterize the set of feasible payoffs given any fixed discount factor of players and the length of interaction. This allows us to obtain the feasible payoff set in closed form. Second, we provide novel comparative statics of the feasible payoff set with respect to the discount factor and the length of interaction. Interestingly, the feasible payoff set becomes \emph{smaller} as players' discount factor becomes larger. Additionally, we identify a necessary and sufficient condition for this monotonicity to be strict. 
\end{abstract}


\strut


\onehalfspacing

\newpage

\renewcommand{\arraystretch}{1.3}

\section{Introduction}

In overlapping generation (OLG) repeated games, players play for finite periods and are replaced by their next generation. This class of games has been used to study cooperation among finitely-lived players in long-run organizations (e.g., \cite{Hammond_1975} and \cite{Cremer_1986_QJE}).

In this paper, we study the feasible payoff set of OLG repeated games. In the literature on OLG repeated games, including studies on folk theorems, the convex hull of the stage game payoffs is typically considered the feasible payoff set of interest. However, the overlapping structure of these games allows players to achieve average discounted payoffs beyond the convex hull of the static payoffs: Although players share the same discount factor, they discount payoffs differently depending on their position in the lifecycle. Thus, it is not obvious which payoffs are feasible. Our purpose in this study is to understand how the OLG structure affects what players can obtain. 

 We study the feasible payoff set when players' discount factor and the period of overlap are fixed, which departs from most studies in the literature, which usually focus on the asymptotic case. Specifically, we focus on ``stationary'' feasible payoffs in which each generation of the same player plays the same sequence of actions during their lifetime.\footnote{In general, each generation of the same player plays different sequences of action profiles. In this case, each player has an infinite sequence of feasible payoffs. We discuss this further in \Cref{subsec:discuss}.}

First, we completely characterize the feasible payoff set of OLG repeated games. We find that the feasible payoff set can be characterized by the convex hull of the set of average discounted payoffs that can be achieved by playing length-$n$ sequences of action profiles, where $n$ is the number of players and each of the action profiles is supposed to be played for $T$ (the interaction length) times consecutively. For such sequences of action profiles, we can calculate the average discounted payoff \emph{as if} players play length-$n$ sequences (rather than $nT$), while effectively discounting by $\delta^T$ (rather than by $\delta$). Thus, this characterization substantially simplifies the set of action profiles we need to consider to obtain the feasible payoff set. In fact, our characterization allows a closed-form expression of the feasible payoff set of any stage game.

Second, we analyze the comparative statics of the feasible payoff set with respect to $\delta$ and $T$. We find that the feasible payoff set is decreasing  (in the set-inclusion sense) in the effective discount factor $\delta^T$. Interestingly, this implies that the set is \emph{decreasing} in $\delta$. When the effective discount factor is $1$, the feasible payoff set coincides with the convex hull of the static payoffs. When it is close to $0$, it is an $n$-dimensional cube, where for each player, the maximum (resp. minimum) feasible payoff coincides with the maximum (resp. minimum) stage game payoff. For intermediate effective discount factors, it is an $n$-dimensional polytope. Furthermore, we examine the condition under which the above monotonicity holds in a strict sense. We find that unless the one-shot feasible set is already a (multidimensional) cube, the OLG feasible set satisfies the strict monotonicity with respect to the parameters.

\textbf{Related Literature} 
	
	Previous research on OLG repeated games has mainly focused on folk-theorem-like approaches such as those in \cite{Kandori_1992_RES} and \cite{Smith_1992_GEB}. Both studies regard the convex hull of the stage game payoffs as the feasible payoff set, which is targeted by their folk theorems. Since their folk theorems first choose sufficiently large $T$ and then $\delta$ close to 1, the feasible set for such a case (as shown in the current study) is the convex hull of the stage game payoffs.\footnote{More precisely, \cite{Kandori_1992_RES} assumes $\delta =1$ throughout the paper, but mentions that the analogous result holds with sufficiently large $\delta<1$ because of the strictness of the punishments. \cite{Smith_1992_GEB} considers two types of folk theorems: non-uniform and uniform. In the non-uniform folk theorem, he first chooses sufficiently large $T$ and then chooses $\delta$. In the uniform folk theorem, he chooses $T$ and $\delta$ simultaneously.} More recently, \cite{Morooka_2021_IJGT} provides an alternative OLG folk theorem with the opposite order of choosing parameters: It shows that if $\delta$ is chosen first and then $T$ is chosen, any feasible and strictly individually rational payoffs can be achieved by subgame perfect equilibrium payoffs, where the feasible payoff set considered is larger than the convex hull of the stage game payoffs. Alternatively, in this study, we analyze OLG repeated games with fixed $\delta$ and $T$. As mentioned, we provide a complete characterization of the feasible payoff set. By studying the feasible payoff set, we provide a natural benchmark for the equilibrium payoff set, which may be of greater interest.

The present study is also related to the literature on repeated games with differential discounting of players, which has been studied since \cite{LP_1999_ECMA}. They study infinitely repeated games between two players with different discount factors and show that some payoffs outside the convex hull of the stage game payoffs can be obtained by intertemporal trading of payoffs, i.e., the more patient player concedes payoffs in early periods to obtain higher payoffs in later periods. \cite{CT_2012_GEB} generalize this analysis to the case of more than two players and provide a folk theorem. \cite{Sugaya_2015_TE} proves a folk theorem for $n$-player infinitely repeated games with imperfect public monitoring. \cite{DG_2022_JET} provide a more constructive approach to studying feasible and equilibrium payoffs of repeated games with perfect monitoring. On the other hand, \cite{Chen_2007_EL} and \cite{CF_2013_IJGT} study finitely repeated games between two players. These papers examine whether the feasible payoff set expands as the length of the game increases. The latter paper, built on the result of the former, shows that this is indeed the case for any two-player stage games.\footnote{They leave the question open for a more general case involving an arbitrary number of players.}

Compared to the literature on repeated games with differential discounting, in our model of OLG repeated games, players share the \emph{same} discount factor. Nevertheless, players can trade payoffs across periods due to the overlapping generation structure: In a given period, players are located in different positions in their lifecycle (``age''), resulting in different discounting of some future payoffs. Notice that players discount their payoffs in the same way when they are the same age. In this sense, their discounting is ``symmetric,'' which makes our analysis relatively tractable. This results in our characterization of the feasible payoff set, allowing a closed-form expression of the feasible payoff set given any stage game for each discount factor and the length of each generation's lifespan.\footnote{In the literature on repeated games with differential discounting, \cite{Chen_2007_EL} provides an explicit characterization of the feasible payoffs of finitely repeated games for a specific two-player stage game. \cite{Sugaya_2015_TE} provides a recursive characterization of the feasible payoff set of infinitely repeated games for general stage games. \cite{DG_2022_JET} provide several characterizations of the feasible payoff set; in particular, they characterize it when players can have large discount factors.}

The remainder of this paper is organized as follows. In \Cref{sec:2}, we introduce the model of OLG repeated games. In \Cref{sec:3}, we present our first main result, which is a complete characterization of the feasible payoff set of OLG repeated games. In \Cref{sec:4}, we provide the comparative statics results of the feasible payoff set with respect to $\delta$ and $T$. In \Cref{sec:5}, we provide two examples to illustrate our main results. In \Cref{sec:6}, we generalize the model and extend our main results. In \Cref{sec:7}, we have some discussions. \Cref{sec:8} concludes the paper. Omitted proofs can be found in the appendix.

\label{sec:1}

\section{Model}

\subsection{Stage Game}
A stage game is defined as a triple $G=(N,(A_i)_i,(u_i)_i)$, where $N=\{ 1,2,\cdots ,n\}$, for some $n\geq 2$, is the set of players, $A_i$ is a finite set of actions available to player $i$, and  $u_i: \prod_{i \in N} A_i \to \mathbb{R}$ is player $i$'s one-shot payoff function. Let $A \equiv \prod_{i \in N} A_i$ be the set of action profiles. 

Given a set $B \subseteq \mathbb{R}^N$, let $co (B)$ be the convex hull of $B$. Let $V \subseteq \mathbb{R}^n$ be the feasible set of one-shot payoffs, defined as: 
$$V \equiv co \left(  \{ u(a) : a \in A\} \right).$$
Let $F^*\equiv \prod _{i\in N} \left [\min _{a\in A}u_i(a),\ \max _{a\in A}u_i(a) \right].$
That is, $F^*$ is the smallest (multidimensional) cube that contains the one-shot feasible set $V$.

\subsection{OLG Repeated Game}
Given a stage game $G$, $\delta \in (0,1]$ and $T \in \mathbb{N}$, we define the \emph{OLG repeated game} $OLG(G,\delta ,T)$ as follows:\footnote{
We generalize the model and extend the main results in \Cref{sec:6}.}
\begin{itemize}
	\item In every period $t \in \mathbb{N}$, $G$ is played by $n$ finitely-lived players.
	\item For $i\in N$ and $d \in \mathbb{N}$, the player with $A_i$ in generation $d$ joins the game at the beginning of period $(d-1)nT+(i-1)T+1$, and lives for the following $nT$ periods until she or he retires at the end of period $dnT+(i-1)T$. The only exceptions are the players with $A_i$ for $i\in N\setminus\{ 1\}$ in generation 0, who participate in the game between periods 1 and $(i-1)T$ (see \autoref{o1}).\footnote{Such an overlapping structure can be found, for instance, in an organization with a fixed retirement age. Employees of the organization work with other employees of different ages. Employees join the organization when they are young. However, at the onset of the organization, some employees are old (generation 0).}
	\item Each player's per-period payoffs are discounted by a common discount factor $\delta$. 
\end{itemize}
\begin{figure}[t]
\centering
\scalebox{0.7}{
\begin{tabular}{|c|c|c|c|c|c|c|c|c|}\hline
Period&$1\sim T$&$T+1\sim 2T$&$2T+1\sim 3T$&$3T+1\sim 4T$&$4T+1\sim 5T$&$5T+1\sim 6T$&$6T+1\sim 7T$&$\cdots$\\\hline
$A_1$&\multicolumn{3}{c|}{Generation 1}&\multicolumn{3}{c|}{Generation 2}&\multicolumn{2}{c|}{Generation 3 $\cdots$}\\\hline
$A_2$&Generation 0&\multicolumn{3}{c|}{Generation 1}&\multicolumn{3}{c|}{Generation 2}&$\cdots$\\\hline
$A_3$&\multicolumn{2}{c|}{Generation 0}&\multicolumn{3}{c|}{Generation 1}&\multicolumn{3}{c|}{Generation 2 $\cdots$}\\\hline
\end{tabular}
}
\caption{Structure of the OLG repeated game with $n=3$}
\label{o1}
\end{figure}

Note that a player interacts with the same opponents (with unchanged generation) for $T$ periods. We refer to such $T$ periods as an \emph{overlap}. For each $i \in N$, we refer to any player whose action set is $A_i$ as ``player $i$." Whenever necessary, we explicitly mention players' generations.

When a sequence of actions $(a (t)) _{t=1}^{nT}\in A^{nT}$ is played throughout a player's life with $A_i$, her/his average payoff is as follows:\footnote{For the player with $A_i$ for $i\in N\setminus\{ 1\}$ in generation 0, replace $nT$ with $(i-1)T$.}
\begin{eqnarray}
\frac{1}{\sum _{t=1}^{nT}\delta^{t-1}}\sum _{t=1}^{nT}\delta ^{t-1}u_i(a (t)).
\nonumber\end{eqnarray}

We maintain the following assumption throughout this paper.\footnote{This assumption is also employed by \cite{Chen_2007_EL} and \cite{CF_2013_IJGT}.} 
\begin{ass}
\label{ass}
Players can access a Public Randomization Device (henceforth PRD) at the beginning of every period. Players observe the realizations of PRDs after their birth. 
\end{ass}

\label{sec:2}

\section{A Complete Characterization of the Feasible Payoff Set}
\label{sec:3}

We are interested in the set of payoffs that can be achieved by a sequence of action profiles in which every generation of players plays the same sequence of action profiles. We call such action sequences and feasible payoffs \emph{stationary}. Stationary action sequences are employed as equilibrium paths for folk theorems \cite[]{Kandori_1992_RES, Smith_1992_GEB}. In this paper, we focus on stationary action sequences and feasible payoffs. Thus, we omit ``stationary'' whenever the meaning is clear. The main result in this section shall provide a complete characterization of (stationary) feasible payoffs for any $\delta$ and $T$.

Since we restrict our attention to stationary sequences, we focus on the payoffs of players in generation $d\geq 1$. For $a^{[nT]}=(a^1,a^2,\cdots ,a^{nT})\in A^{nT}$, we define the average discounted payoff $U_i(a^{[nT]})$ of player $i$ as follows:
\begin{equation}
\label{eq:nnnn1}
U_i(a^{[nT]}):=\frac{1}{\sum _{k=1}^{nT}\delta ^{k-1}}\left(\sum _{k=1}^{nT}\delta ^{k-1}u_i(a^{(i-1)T+k}) \right),	
\end{equation}
where $a^{s}=a^{s-nT}$ for $s\geq nT+1$. Let $U( a^{[nT]}) \equiv (U_i (a^{[nT]}))_{ i\in N}.$

We define the \emph{(stationary) feasible payoff set} given $\delta$ and $T$ by 
$$
F(\delta,T):=co \left(\{ U(a^{[nT]}):a^{[nT]}\in A^{nT}\} \right).
$$
Since we are interested in stationary feasible payoffs, the definition would be the most permissible one, as it allows any correlation over length-$nT$ action sequences.

It is convenient to introduce an additional notation. Given $a^{[n]}=(a^1,a^2,\cdots ,a^n)\in A^n$ and $i\in N$, we define the value $v_i(a^{[n]})$ as follows: 
\begin{align}
v_i(a^{[n]})&:=\frac{1}{\sum _{k=1}^n\delta ^{(k-1)T}}\left(\sum _{k=1}^n\delta ^{(k-1)T}u_i(a^{i+k-1}) \right), \label{eq:nnn1}
\end{align}
where $a^{s}=a^{s-n}$ for $s\geq n+1$. Let $v( a^{[n]}) \equiv (v_i (a^{[n]}))_{ i\in N}.$ That is, $v_i (a^{[n]})$ represents the average discounted payoff of player $i$ from repeatedly playing the same action profile $a^k$ during the $k$-th overlap:
$$U_i ( \underbrace{a^1, \dots, a^1}_{\text{$T$ times}}, \dots, \underbrace{a^n, \dots, a^n}_{\text{$T$ times}}) = v_i (a^1,  \dots, a^n).$$
We call such sequences of action profiles \emph{stable}.
Alternatively, $v_i ( a^{ [n]})$ can be interpreted as the average discounted payoff over $n$ periods, where the ``effective'' discount factor is $\delta^T$. 

We are now ready to present our first main result.
\begin{them} 
\label{them:1}
For any $\delta \in (0,1]$ and $T \in \mathbb{N}$, 
\begin{equation}
\label{eq:nn1}
F(\delta,T) =co\left( \left \{v(a^{[n]}): {a^{[n]}\in A^n} \right \} \right).
\end{equation}

\end{them}

In words, the feasible payoff set coincides with the convex hull of the average discounted payoffs of length-$n$ sequences of action profiles (i.e., stable paths), each of which is to be played $T$ times. Notice that we could interpret the average discounted payoff of such a sequence of action profiles as the average discounted payoff of a length-$n$ sequence of action profiles, discounted by $\delta^T$. 

The characterization is useful since it substantially reduces the number of sequences of action profiles we need to consider: Regardless of $T$, it is sufficient to consider length-$n$ sequences. For instance, for the Prisoners' Dilemma game, where each player has two actions (so four action profiles), the result implies that it is sufficient to consider $4 \times 4 = 16$ sequences of action profiles to obtain the feasible payoff set, regardless of $\delta$ and $T$. The result also means that $\delta$ and $T$ affect the feasible payoff set only through $\delta^T$. Thus, any $(\delta, T)$ and $(\delta', T')$ with $\delta^T ={\delta'}^{T'}$ would result in the same feasible payoff set.

One may wonder whether the right-hand side of \eqref{eq:nn1} suggests that the feasible payoffs may not be implemented under our assumption regarding the observability of PRDs. That is, we assumed that a result of a PRD is observed only by the contemporary generation, whereas for each $T$ period (i.e., overlap), there is a player to be replaced; thus, it might require a stronger informational assumption, for instance, the results of PRDs are also observed by future generations. In the proof, however, we construct a sequence of PRDs, each of which is executed every overlap that generates the same average discounted payoff from a sequence of action profiles.

The underlying idea of the proof is as follows: Since no player retires within an overlap, every player discounts their payoff at $t+1$ relatively more than the payoff at $t$ in the overlap by $\delta$, although they discount payoffs differently depending on their ``age.'' This allows us to generate the same average discounted payoff for a given sequence of action profiles for a given overlap using a PRD at the beginning of this overlap, where the probability of playing a certain action profile during this overlap is determined in a way to mimic the sum of discounted payoffs whenever this action profile is played in the original sequence of action profiles. On the other hand, across overlaps, some player should retire, and for such a player, the relative discounting between $t$ and $t+1$ is different, which makes a similar construction of a PRD between overlaps unavailable.

\begin{proof}

We first show $ F(\delta,T) \subseteq co\left(\{v(a^{[n]}): {a^{[n]}\in A^n}\} \right) .$
Since $co\left(\{v(a^{[n]}): {a^{[n]}\in A^n}\} \right)$ is a convex set, it is sufficient to show that for an arbitrary sequence $a^{[nT]} = (a^1, \dots, a^{nT}) \in A^{nT}$, the average discounted payoff from this sequence can be obtained by a convex combination of elements in $\{v(a^{[n]}): {a^{[n]}\in A^n}\}$.  

We first claim that for each overlap $k=1, \dots, n$, the average discounted payoff within the overlap can be achieved by playing constant action profiles resulting from a PRD. To see this, consider the following PRD, $\alpha^k \in \Delta (A)$, for the $k$-th overlap, defined as follows:
$$\alpha^k (a) := \sum_{t=(k-1)T+1}^{kT} \frac{\delta^{t-1 - (k-1)T}}{ 1+ \dots + \delta^{T-1}}\mathbf{1}_{ \{{a}^t = {a} \}}, \quad \forall a \in A.$$
Clearly, $\sum_{ a \in A} \alpha^k ({a}) = 1$. The PRD is executed at the beginning of the overlap, and the players are supposed to play the realized action profile consecutively until the end of the overlap. To see that this PRD generates the same average discounted payoff as that of $a^{[nT]}$ for each player during the overlap, note that the PRD is constructed in such a way that it mimics discounting. Namely, if $a^t = a$ in $a^{ [nT]}$ for some $t  \in \{ 1, 2, \dots, T\}$, then the probability of playing $(a,\dots, a)$ increases by $\frac{\delta^{t-1- (k-1)T}}{\sum_{t=1}^T \delta^{t-1} }$. This construction is possible because no player retires within the overlap so that every player discounts the $t+1$-period payoff by $\delta$ times more than the $t$-period payoff. This is not the case as long as some player retires and is reborn (i.e., across overlaps). 

Lastly, we observe that the average discounted payoff of each player from this sequence of PRDs $(\alpha^1, \dots, \alpha^n)$ is a convex combination of $\{v(a^{[n]}):a^{[n]}\in A^n\}$, where the weight for each $(a^1, \dots, a^n ) \in A^n$ is $\alpha^1 (a^1)\times \cdots \times \alpha^n (a^n)$.

For the converse, observe that each element of $\{v(a^{[n]}):a^{[n]}\in A^n\}$ is in $F(\delta, T)$, and $F(\delta, T)$ is convex by definition. 
\end{proof}

\section{Comparative Statics of the Feasible Payoff Set}
\label{sec:4}

In this section, we study the comparative statics of the feasible set with respect to both $T$ and $\delta$. From \autoref{them:1} in the previous section, we know that the feasible payoff set depends on $\delta$ and $T$ only through $\delta^T$.

\subsection{Monotonicity}
Our main result in this section asserts that the feasible payoff set is non-increasing in $\delta^T$.

\begin{them}
\label{them:2}
For any $\delta, \delta' \in (0,1]$ and $T, T' \in \mathbb{N}$, if $\delta^T \leq  \delta{'}^{T'}$, $F(\delta', T') \subseteq F(\delta, T)$. In particular, the following hold:
\begin{enumerate}
	\item Given any $T \in \mathbb{N}$, $F(\delta ',T)\subseteq F(\delta ,T)$ for any $\delta, \delta'$ with $0< \delta <\delta ' \leq 1$. 
\item Given any $\delta \in (0,1]$, $F(\delta , T)\subseteq F(\delta ,T+1)$ for any $T \in \mathbb{N}$.	
\end{enumerate} \end{them}

Payoffs outside the convex hull of the stage game (i.e., $V$) are generated by trading payoffs across different generations. Such intertemporal trading can be beneficial because players have different ``ages,'' meaning they weigh the payoffs differently, even if their discount factor $\delta$ is the same. This means that one player's loss can differ from another's gain. A larger $T$ allows more interaction and more trading opportunities among the players. On the other hand, when $\delta$ increases with fixed $T$, all players put similar weights on all periods, so a player's gain is approximately equal to another player's loss. Hence, the feasible set shrinks as $\delta$ increases to 1 with $T$ fixed.

The rest of this section is devoted to proving \autoref{them:2}.\footnote{The weak set inclusions in the result cannot be strengthened to be strict. For instance, if the stage game payoff set is an $n$-dimensional cube, the feasible payoff set of the corresponding OLG repeated game would be unchanged with respect to the parameters.} In doing so, we shall also provide a characterization of players' optimal play to maximize the welfare given some weights for players, which might have some independent interest.

Thanks to \autoref{them:1}, we consider a sequence of $n$ pure action profiles without loss of generality to study the monotonicity. That is, players play the same action profile during an overlap that consists of $T$ periods. Additionally, since the feasible payoff set is convex, it is enough to show that the maximum ``score'' increases as $\delta$ (resp. $T$) becomes smaller (resp. larger) for each non-zero direction. 

Fix $\lambda \in \mathbb{R}^n \setminus \{ \mathbf{0}\}$. For a given $\delta \in (0,1]$ and $T \in \mathbb{N}$, let $\Delta \equiv  \Delta (\delta, T) =  \delta^T$.  Define $\lambda$-weighted welfare as:
\begin{equation}
\label{eq:nn2}
W_{ \lambda}^*(\Delta) :=\max_{ a^{[n]} =(a^1, \dots, a^n) \in A^n} W_\lambda (a^{[n]},\Delta),
\end{equation}
where
$$W_\lambda (a^{[n]}, \Delta ) := \sum_{i=1}^n \lambda_i v_i (a^1, \dots, a^n), \quad \forall a^{[n]} \in A^n$$
and $v_i$ is defined in \eqref{eq:nnn1}.\footnote{$W_\lambda$ becomes a function of $\Delta$ because $v_i$ depends on $\delta$ and $T$ only through $\Delta$.}

We want to show that $W_{ \lambda}^* (\Delta)$ is non-increasing in $\Delta$ (consequently, non-decreasing in $T$ and non-increasing in $\delta$).

We introduce a few more notations. Denote the set of optimal solutions of \eqref{eq:nn2} by $\mathcal{A}^*_\lambda (\Delta) \subseteq A^n$. Let $\mathcal{U}^*_\lambda (\Delta) \equiv \left \{ (u^k)_k : u^k = u (a^k), a^{[n]} \in \mathcal{A}_\lambda^* (\Delta) \right \} \subseteq \mathbb{R}^{n\times n}$, i.e., the set of sequences of the optimal payoff vectors. Notice that $\mathcal{U}^*_\lambda (\Delta)$ may be a singleton even when $\mathcal{A}^*_\lambda (\Delta)$ is not.  Lastly, given $u^{[n]} = (u^1, \dots, u^n) \in \mathbb{R}^{n \times n}$, for each $k \in \{ 1, \dots, n\}$, let
$$w_k (u^{[n]}) := \sum_{ i =1}^n \lambda_i u_i^{i+k-1},$$
where $u_i^m=u_i^{m-n}$ if $m>n$. That is, $w_k (u^{[n]})$ is the weighted sum of players' payoffs when their ``age'' is $k$ (i.e., they are in the $k$-th overlap in their lifetime). For instance, given $u^{[3]} = (u^1, u^2, u^3) \in \mathbb{R}^{3 \times 3}$ and $\lambda = (1,1,1)$, $w_1 (u^{[3]}) = u^1_1  + u^2_2   + u^3_3 $, $w_2 (u^{[3]}) = u^2_1  + u^3_2  + u^1_3  $, and $w_3 (u^{[3]}) =u^3_1   + u^1_2  + u^2_3 $.

Observe that for each $u^{ [n]}=(u^1,\dots, u^n) \in \mathcal{U}^*_\lambda (\Delta)$, 
\begin{equation}
\label{eq:nn3}
W_{ \lambda}^* (\Delta) = \frac{\sum_{k=1}^n \Delta^{k-1} w_k (u^{[n]}) }{\sum_{k=1}^n \Delta^{k-1}}.
\end{equation}

The following lemma is the key to proving the monotonicity. 
\begin{lem}
\label{lem:1}
Let $\Delta \in (0,1]$ and $u^{[n]}\in \mathcal{U}^*_\lambda (\Delta)$. For each $m =1, \dots, n-1$, 
\begin{equation}
\label{eq:2}
\sum_{k=1}^m \Delta^{k-1} w_k (u^{[n]}) \geq  \sum_{k=1}^m \Delta^{k-1} w_{n-m+k} (u^{[n]}).
\end{equation}
\end{lem}

In words, the lemma means that, at optimum, the payoffs in the earlier stages of a player are higher than those in the later stages (in the sense that for any $m \leq n-1$, the discounted sum of the first $m$ payoffs should be larger than that of the last $m$ payoffs). When there are only two players, this reduces to the condition that the payoffs when they are ``young'' must be larger than those when they are ``old'' for optimality. When there are more than two players, it is \emph{a priori} not clear what the corresponding expression could be. According to the lemma, for the case of three players, it is $w_1 (u^{[3]}) \geq w_3 (u^{[3]}) $ and $w_1 (u^{[3]})   + \Delta w_2 (u^{[3]})  \geq w_2 (u^{[3]})  + \Delta w_3 (u^{[3]}) $.

Let us explain the crux of the idea of the proof of this lemma with three players and $\lambda = (1,1,1)$. Assume $u^{[3]} = (u^1, u^2, u^3)$ is an optimal sequence of payoff vectors. From the optimality of $u^1$ at overlap $k=1$, we have $u^1_1 + \Delta^2 u^1_2 + \Delta u^1_3 \geq u^2_1 + \Delta^2 u^2_2+ \Delta u^2_3$. By multiplying both sides by $\Delta$, 
		$$\Delta u^1_1 + \Delta^3 u^1_2 + \Delta^2 u^1_3 \geq \Delta u^2_1 + \Delta^3 u^2_2 + \Delta^2 u^2_3.$$
On the other hand, by the optimality of $u^2$ at overlap $k=2$,
		$$\Delta u^2_1 + u^2_2 + \Delta^2 u^2_3\geq \Delta u^1_1 + u^1_2 + \Delta^2 u^1_3.$$
		From the two inequalities, we can conclude that $u^2_2 \geq u^1_2$. Intuitively, since the same generation of player 1 and 3 is active both at $k=1$ and $k=2$, while player 2 is replaced by the next generation, in order for $u^2$ to give a larger aggregate payoff at $k=2$, player 2 should have a larger payoff in $u^2$ than in $u^1$. A symmetric argument results in $u^3_3 \geq u^2_3$ and $u^1_1 \geq u^3_1$, and summing them up results in $w_1 (u^{[3]}) \geq w_3 (u^{[3]}) $. Applying a similar argument to the case of ``two-overlap apart,'' we have $\Delta u_2^1 + u_3^1 \leq \Delta u_2^3 + u_3^3$ (as a result, $w_1 (u^{[3]})   + \Delta w_2 (u^{[3]})  \geq w_2 (u^{[3]})  + \Delta w_3 (u^{[3]})$).

The following lemma states that the inequality \eqref{eq:2} in \autoref{lem:1} is sufficient to prove that the derivative of the aggregate payoff with respect to $\Delta$ is non-positive. 
\begin{lem}
\label{lem:2}
For $\Delta \in (0,1)$,
\begin{equation}
\label{eq:nnnn7}
\frac{\partial W_{ \lambda} }{\partial \Delta }(a^{[n]},\Delta)  \leq 0
\end{equation}
for any $a^{[n]}  \in \mathcal{A}_\lambda^* (\Delta)$.
\end{lem}

In the proof of this lemma, we show that the numerator of the derivative in \eqref{eq:nnnn7} is a weighted sum of $m-1$ terms, each of which is $(RHS-LHS)$ for $m=1, \dots, n-1$ in \autoref{lem:1}. This immediately implies the inequality in \eqref{eq:nnnn7}.

\begin{proof}[Proof of \autoref{them:2}]

We first observe that $W_\lambda^*(\Delta)$ is continuous in $\Delta$ everywhere because it is the maximum of continuous functions. 

In addition, we show that $W_\lambda^*(\Delta)$ is differentiable at all but finitely many $\Delta$. Observe that for any $a^{[n]}, \tilde{a}^{[n]} \in A^n$ such that $W_\lambda (a^{[n]}, \Delta)  \neq  W_\lambda (\tilde{a}^{[n]}, \Delta) $ for some $\Delta$, the number of $\Delta$ that solves $W_\lambda (a^{[n]}, \Delta) - W_\lambda (\tilde{a}^{[n]}, \Delta)=0$ is at most $n-1$, because it is a polynomial equation of degree at most $(n-1)$. Thus, the total number of intersections of $W_\lambda ( \cdot, \Delta)$ that all such pairs of action sequences can make is at most $\frac{|A^n||A^n -1|}{2} (n-1)$. This implies that we can find a partition $\{ (\underline{b}_l,\overline{b}_l]: l = 1,\dots, L \}$ of $(0,1]$, where $L \in \mathbb{N}$, such that for each $l = 1, \dots, L$, there is a unique payoff sequence $u_l^{[n]} = (u(a_l^k))_{k=1}^n$ corresponding to a solution $a_l^{[n]} =(a_l^k)_{k=1}^n$ which is optimal for any $\Delta \in (\underline{b}_l,\overline{b}_l]$. Hence, $W_\lambda^* (\Delta)$ is differentiable in the interior of $(\underline{b}_l,\overline{b}_l]$.\footnote{For $\Delta$ on the boundary, the derivative may not exist. For instance, in the OLG repeated game with the Prisoners' Dilemma depicted in \autoref{fig:1}, for large $\Delta$ and $T=1$, when the direction is $\lambda = (1,1)$, $(CC, CC)$ is optimal, while for some small $\Delta$, the optimal sequence is $(DC,CD)$. There is a cutoff $\hat{\Delta}$ for this change, and neither $(CC,CC)$ nor $(DC,CD)$ satisfies the description above: $(CC, CC)$ (or the corresponding payoff sequence) remains optimal for $[\hat{\Delta}, \hat{\Delta} + \epsilon)$, while $(DC, CD)$ remains optimal for $(\hat{\Delta} - \epsilon, \hat{\Delta}]$ for some small $\epsilon>0$.
 } Further, the derivative is non-positive by \autoref{lem:2}, i.e., 
\begin{equation*}
\label{eq:nn6}
\frac{d W^*_\lambda (\Delta) }{d\Delta } = \frac{\partial W_{ \lambda} }{\partial \Delta }(a_l^{[n]},\Delta) \leq 0.
\end{equation*}
Together with the continuity of $W^*_\lambda (\Delta)$, this implies its monotonicity. 
\end{proof}

\subsection{A Limit Characterization of the Feasible Payoff Set}

For the asymptotic cases (i.e., $\delta^T \nearrow 1$ or $\delta^T \searrow 0$), we can be more explicit about the shape of the feasible payoff set. 

When the players' effective discount factor $\delta^T$ is close to 1, there is not much difference in the discounting of young and old players. Therefore, the scope of intertemporal trade of payoffs is limited, and what is the best is to maximize the stage game payoffs for a given welfare weight. On the other hand, when the effective discount factor is close to 0, the difference in the discounting between young and old players is large. Hence, intertemporal trading can be significantly beneficial. The extreme form of such trading is to maximize the youngest player's (weighted) payoff. We summarize this discussion as a corollary:

\begin{coro}
The following hold:
\begin{enumerate}
\item For $\delta^T$ sufficiently close to 1, the solution of the optimization problem \eqref{eq:nn2} is to play some $a^k \in \argmax_{a' \in A} \lambda \cdot u(a')$ for each overlap $k$. Consequently, $\lim_{ \delta^T \nearrow 1} F(\delta, T) = V.$
	\item For $\delta^T$ sufficiently close to 0, the solution of the optimization problem \eqref{eq:nn2} is to play $a^k \in \argmax_{a' \in A} \lambda_{i_k} u_{i_k} (a')$, where $i_k \in N$ is the youngest player in overlap $k$. Consequently, $\lim _{\delta^T\searrow 0}F(\delta, T)=  \prod _{i\in N} \left [\min _{a\in A}u_i(a),\ \max _{a\in A}u_i(a) \right]$.	 	
\end{enumerate}

\end{coro}

Thus, for any $\delta \in (0,1)$, as $T \to \infty$, the feasible set converges to the $n$-dimensional cube. Similarly, for any $T$, as $\delta \to 0$, the feasible set converges to the $n$-dimensional cube.

\subsection{Strict Monotonicity}

In this subsection, we identify a necessary and sufficient condition for the strict expansion of the feasible set as $\delta^T$ becomes smaller.

We show that when $V\neq F^*$, the OLG feasible set satisfies the following strict monotonicity:
\begin{prop}
\label{prop:nnn1}
Suppose $V\neq F^*$. Then, for any $\delta, \delta' \in (0,1]$ and $T, T' \in \mathbb{N}$, if $\delta^T<\delta{'}^{T'}$, $F(\delta', T') \subsetneq F(\delta,T)$. Conversely, if $V = F^*$, then $F(\delta, T)=V$ for any $\delta \in (0,1], T \in \mathbb{N}$.
\end{prop}

That is, unless the one-shot feasible set is already a (multidimensional) cube, the OLG feasible set satisfies the strict monotonicity with respect to $\delta^T$.

Intuitively, if $\delta^T <1$ and $V \neq F^*$, players can find some opportunities of intertemporally trading payoffs (i.e., playing a non-constant action profile sequence) so that they can achieve payoffs beyond $V$ in some direction $\lambda$. It turns out that the benefit from the most efficient trading (i.e., that achieves $W^*_\lambda (\Delta)$) is strictly larger with smaller $\Delta$.

\section{Examples}
\label{sec:5}

In this section, we illustrate our main results in the previous section using well-known stage games in the repeated game literature.  

\subsection{The OLG Prisoners' Dilemma}
\label{subsec:pd}

\begin{figure}[t]
\centering
\begin{tabular}{|c|c|c|}
\hline
    & $C$      & $D$      \\ \hline
$C$ & $1,1$    & $-1,2$ \\ \hline
$D$ & $2,-1$ & $0,0$    \\ \hline
\end{tabular}
\caption{The Prisoners' Dilemma}
\label{fig:1}
\end{figure}
\begin{figure*}
\centering
\includegraphics[width=0.45\textwidth]{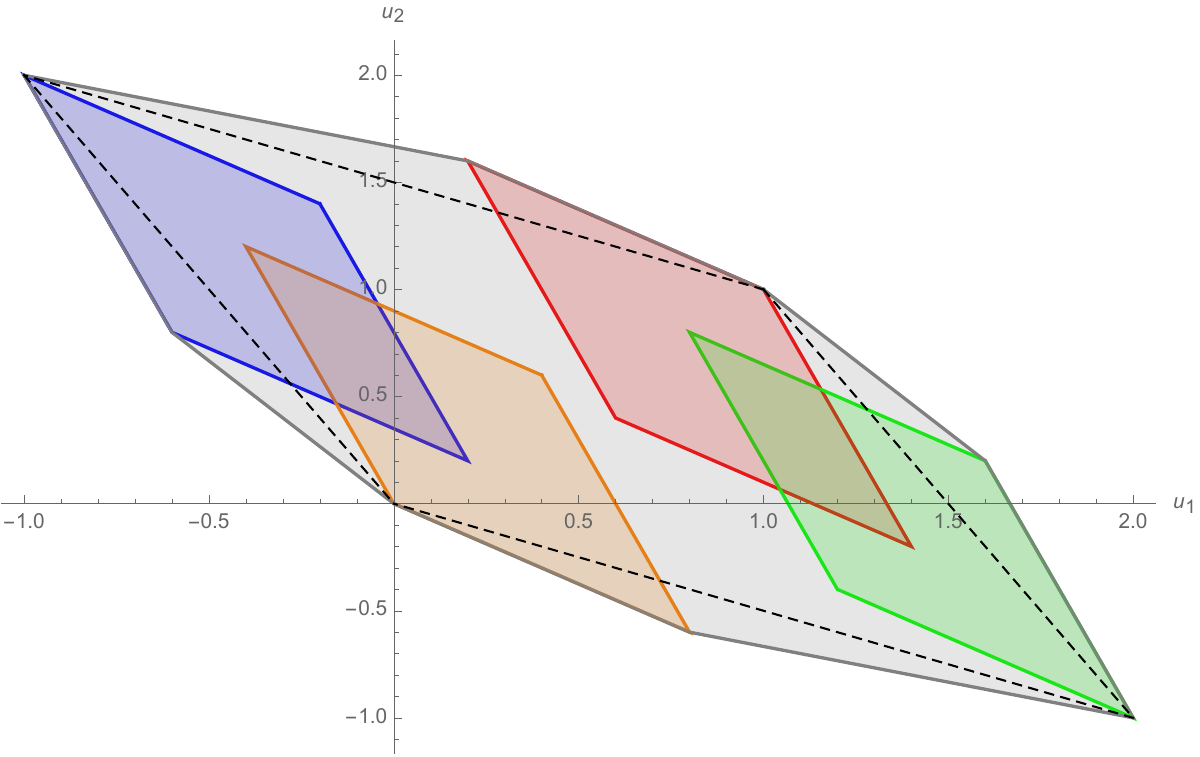}
\includegraphics[width=0.45\textwidth]{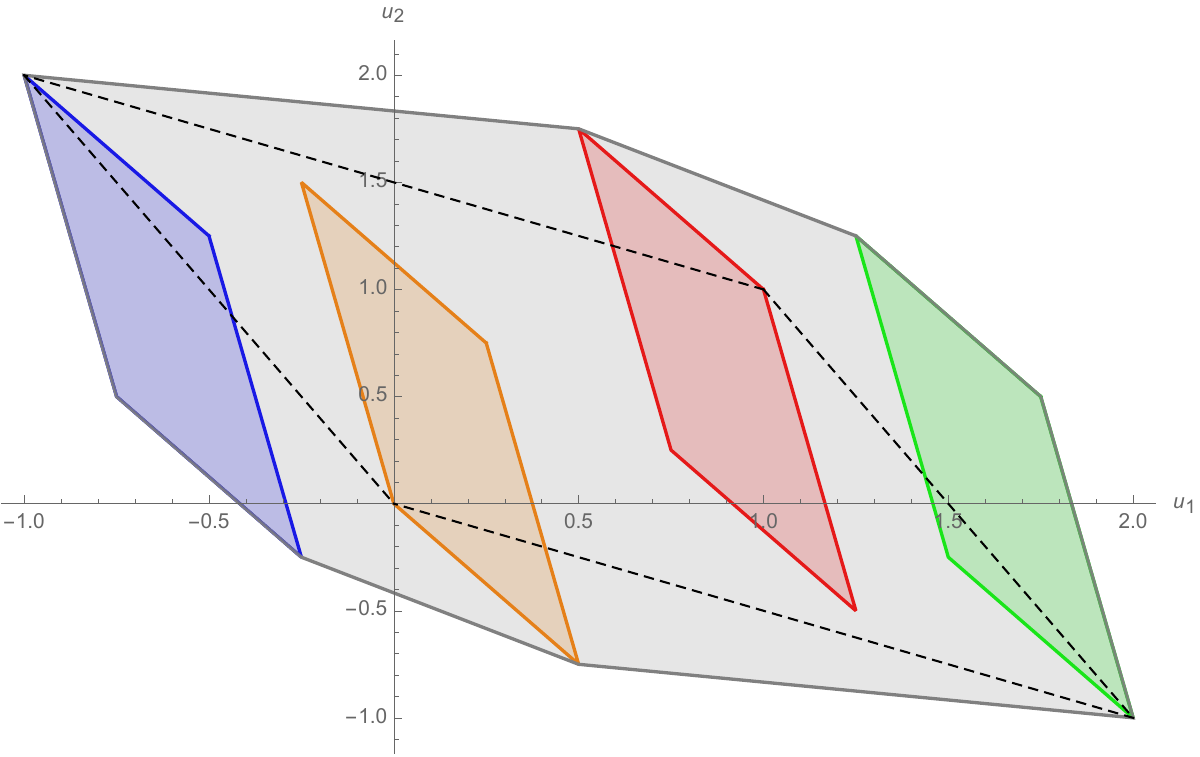}
\includegraphics[width=0.45\textwidth]{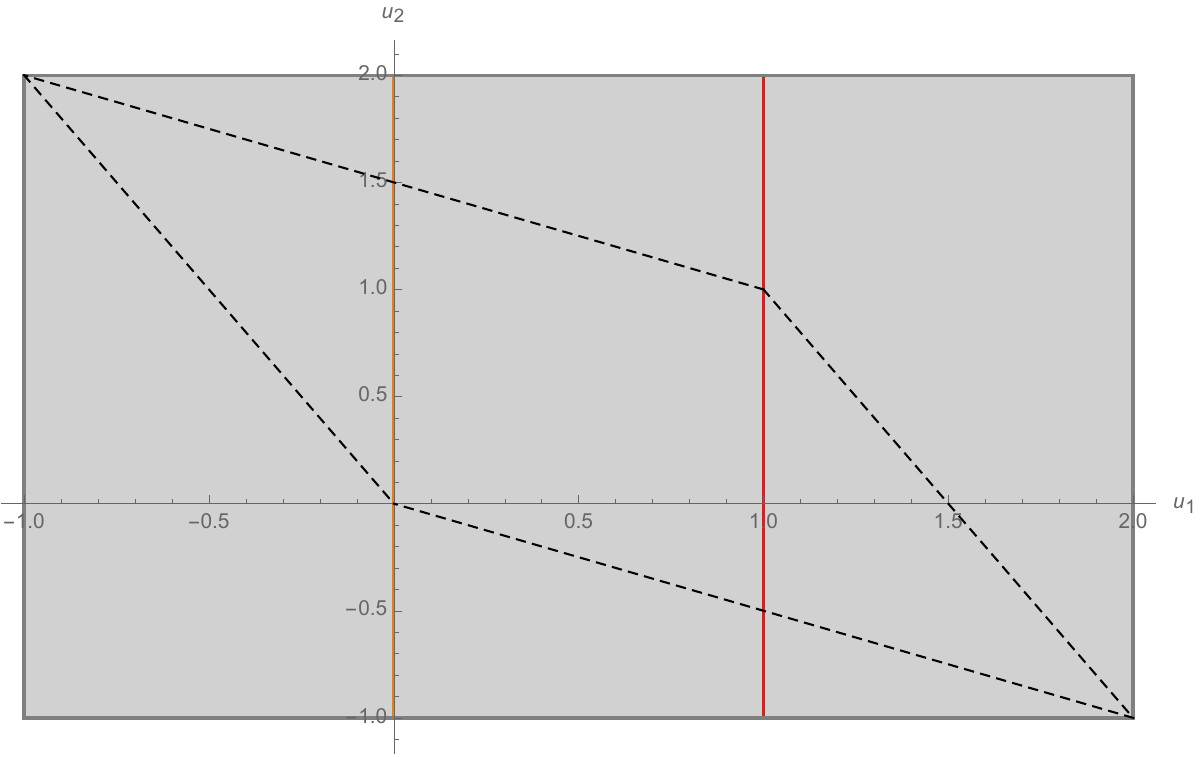}
	\caption{In each figure, the Gray region represents the feasible payoff set (for $\Delta =\frac{2}{3}$, $\Delta = \frac{1}{3}$ and $\Delta \to 0$ clockwise) of the OLG Prisoners' Dilemma game. In each figure, the region surrounded by the dotted lines is the convex hull of the stage game payoffs; the Red region represents the convex hull of the four payoffs, $v(CC, a^2)$, $a^2 \in \{ C,D\}^2$ in \eqref{eq:n2}. Similarly, the Green, Orange, and Blue regions represent the counterparts for $DC, DD$, and $CD$, respectively.}
	\label{fig:n1}
\end{figure*}

Consider the OLG repeated game $OLG(G,\delta ,T)$, where the stage game $G$ is the Prisoners' Dilemma (see \autoref{fig:1}). By \autoref{them:1}, the OLG feasible payoff set is:
\begin{equation}
\label{eq:n2}
	F(\delta, T) = co  \left( \bigcup_{ a \in \{CC, DC, DD, CD\}}  \{ v( a, CC), v( a, DC), v(a, DD), v(a, CD)  \}  \right ), 
\end{equation}
where $v (a^1, a^2) = \left (\frac{u_1 (a^1)  + \Delta u_1 (a^2) }{1+ \Delta}, \frac{\Delta u_2 
(a^1)  + u_2 (a^2)}{1 + \Delta } \right)$ for any $(a^1, a^2) \in A^2$ as previous and $\Delta = \delta^T$ (see \autoref{fig:n1}).

As $\Delta$ becomes smaller (either because players interact for a longer period or because they discount more), players put more weight on the payoffs obtained when they are young. Each colored region in the figure corresponding to a certain $a^1 \in A$ represents the payoff vectors that can be achieved from playing $a^1$ over the first overlap and various $a^2 \in \{CC, DC, DD, CD \}$ over the second overlap. The effect of a smaller $\Delta$ is represented by each region becoming horizontally narrower as player 1's payoff of the four payoff vectors becomes close to $u_1 (a^1)$, and vertically more spread out as player 2's payoff becomes close to $u_2(a^2)$. Thus, in the limit of $\Delta$ converging to $0$, each region becomes a vertical line and their convex combination becomes the rectangle. 

It is notable that as $\Delta$ changes, the sequences of action profiles that achieve the extreme points of the feasible payoff set may change. For instance, when $\Delta$ is large enough (e.g., $\Delta = 2/3$), $(CC, CC)$ yields an extreme point. Thus, no intertemporal trading is needed to achieve the payoff. As $\Delta$ becomes smaller (e.g., $\Delta = 1/3$), $(CC, CC)$ no longer generates an extreme point, while $(DC, CD)$ becomes a new sequence corresponding to one of the extreme points. Intuitively, as $\Delta$ becomes smaller, the benefit of intertemporal trading, such as $(DC,CD)$, becomes larger. Note that for sufficiently small $\Delta$, each player should play the action profile that maximizes her/his payoff in order to be on the efficient frontier.

\subsection{A Three-player Example from \cite{FM_1986_ECMA}}
\label{subsec:fm}
Consider another example involving three players. The stage game is described in \autoref{fig:2}. The feasible payoff set of the standard infinitely repeated game with this stage game is equal to $V$, which is the line segment between $(0,0,0)$ and $(1,1,1)$. Note that the dimension of $V$ (i.e., 1) is less than the number of players (i.e., 3).\footnote{\cite{FM_1986_ECMA} use this stage game to show that a folk theorem may fail for standard repeated games with infinitely-lived players if the stage game does not satisfy the full dimensionality, which is a sufficient condition for their folk theorem. \cite{Smith_1992_GEB} shows that the full-dimensionality is not necessary for his folk theorem for OLG repeated games. Since his folk theorem first chooses $T$ and then chooses sufficiently large $\delta$, it concerns the case when $\delta^T$ is close to 1. The feasible payoff set in this case is the ``smallest,'' according to our characterization, which is the line segment between $(0,0,0)$ and $(1,1,1)$.}
\begin{figure}
\centering
\begin{tabular}{|c|c|c|}
\hline
    & $A$     & $B$     \\ \hline
$A$ & $1,1,1$ & $0,0,0$ \\ \hline
$B$ & $0,0,0$ & $0,0,0$ \\ \hline
\end{tabular}
\quad 
\begin{tabular}{|c|c|c|}
\hline
    & $A$     & $B$     \\ \hline
$A$ & $0,0,0$ & $0,0,0$ \\ \hline
$B$ & $0,0,0$ & $1,1,1$ \\ \hline
\end{tabular}
\caption{The stage game of a three-player pure coordination game}	
\label{fig:2}
\end{figure}

On the other hand, the feasible payoff set of the OLG repeated game features the full dimensionality (i.e., it is equal to the number of players) for any $\Delta \in (0,1)$. When $\Delta = 1$, it is the line segment between $(0,0,0)$ and $(1,1,1)$, which coincides with the convex hull of the stage game payoffs. In contrast, when $\Delta \in (0,1)$, it is a polytope with a nonempty interior. 
\begin{figure}[t]
\centering
\includegraphics[width=0.5\textwidth]{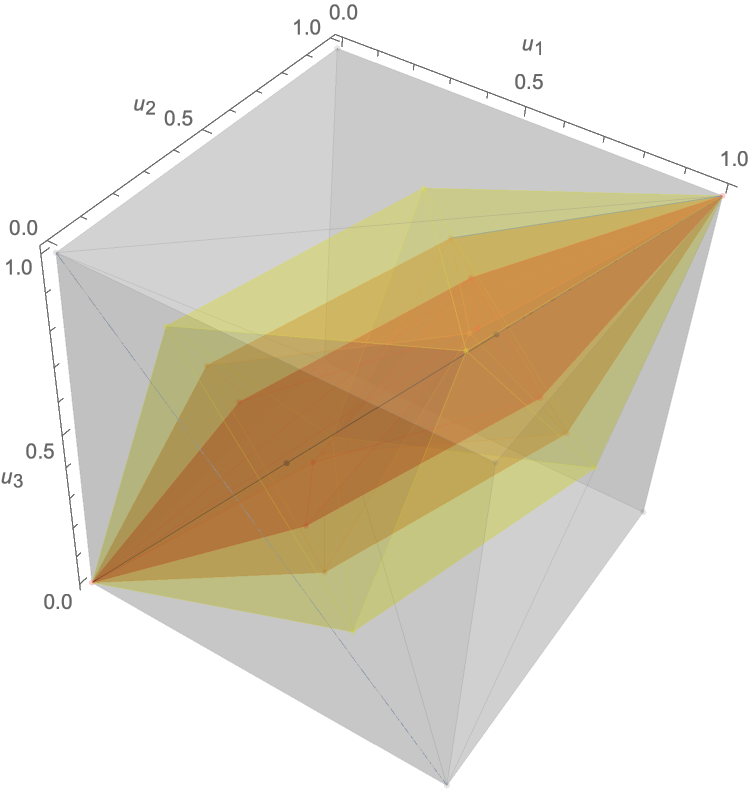}
\caption{The OLG feasible payoff set of the pure coordination game when $\Delta = 1$ (Black), $\Delta = 2/3$ (Dark Orange), $\Delta = 1/2$ (Light Orange), $\Delta = 1/3$ (Yellow), and $\Delta \to 0$ (Gray).}
\label{}
\end{figure}

\section{Extension}
\label{sec:6}

In this section, we generalize the model in \Cref{sec:2} and extend our main results in \Cref{sec:4}.

Let $G = (N,(A_i)_i,(u_i)_i)$ be a stage game, $\delta \in [0,1]$ and $T \in \mathbb{N}$. Also, let $\mathbf{M}=(M_1,M_2,\cdots ,M_n)\in\mathbb{N}^n$ be an \emph{overlapping structure}. An \emph{OLG repeated game with overlapping structure $\mathbf{M}$}, $OLG(G,\delta,T,\mathbf{M})$, is defined as follows:\footnote{\cite{Kandori_1992_RES} studies this model and provides a folk theorem.} 
\begin{itemize}
	\item 
For $i\in N$ and $d\in\mathbb{N}$, the player with $A_i$ in generation $d$ joins the game at the beginning of period $(d-1)\overline{M}_nT+\overline{M}_{i-1}T+1$, and lives for the following $\overline{M}_n T$ periods, until she or he retires at the end of period $d\overline{M}_nT+\overline{M}_{i-1}T$, where $\overline{M}_0=0$ and $\overline{M}_i=\sum_{j=1}^iM_j$ for $i\in N$. The only exceptions are the players with $A_i$ for $i\in N\setminus\{ 1\}$ in generation 0, who participate in the game between periods 1 and $\overline{M}_{i-1}T$.
\item Each player's per-period payoffs are discounted by a common discount factor $\delta$.
\end{itemize}

\begin{figure}[t]
\centering
\scalebox{0.5}{
\begin{tabular}{|c|c|c|c|c|c|c|c|c|}\hline
Period&$1\sim \overline{M}_1T$&$\overline{M}_1T+1\sim \overline{M}_2T$&$\overline{M}_2T+1\sim \overline{M}_3T$&$\overline{M}_3T+1\sim \overline{M}_3T+\overline{M}_1T$&$\overline{M}_3T+\overline{M}_1T+1\sim \overline{M}_3T+\overline{M}_2T$&$\overline{M}_3T+\overline{M}_2T+1\sim 2\overline{M}_3T$&$2\overline{M}_3T+1\sim 2\overline{M}_3T+\overline{M}_1T$&$\cdots$\\\hline
$A_1$&\multicolumn{3}{c|}{Generation 1}&\multicolumn{3}{c|}{Generation 2}&\multicolumn{2}{c|}{Generation 3 $\cdots$}\\\hline
$A_2$&Generation 0&\multicolumn{3}{c|}{Generation 1}&\multicolumn{3}{c|}{Generation 2}&$\cdots$\\\hline
$A_3$&\multicolumn{2}{c|}{Generation 0}&\multicolumn{3}{c|}{Generation 1}&\multicolumn{3}{c|}{Generation 2 $\cdots$}\\\hline
\end{tabular}
}
\caption{Structure of generalized OLG repeated game with $n=3$}
\label{aaaa}
\end{figure}

Note that for each $i \in \{ 1, \dots, n-1\}$, player $(i +1) $ is born later than player $i$ by $M_i \times T$ periods. Let us call such $M_i \times T$ periods an ``overlap.'' Thus, in this extension, overlaps may have different lengths. Note that the model in \Cref{sec:2} is a special case of this generalized model when $\mathbf{M}= (1, \dots, 1) \in \mathbb{R}^n$.\footnote{The generalized model is, however, not fully general. For example, suppose that the stage game has three players. We can consider an OLG repeated game where player 1 in each generation is born in an odd period and dies in the next period, whereas players 2 and 3 in each generation together are born in an even period and die in the next period. This model belongs to one structure of OLG repeated games. However, the OLG model in the present study does not include such a model.}

We define the feasible payoff set of $OLG (G, \delta, T,\mathbf{M})$ by 
$$
F(\delta,T, \mathbf{M}):=co \left(\{ U^{\mathbf{M}}(a^{[\overline{M}_nT]}):a^{[\overline{M}_nT]}\in A^{\overline{M}_nT}\} \right),
$$
where $U^{\mathbf{M}}: A^{\overline{M}_nT} \to \mathbb{R}^n$ is defined similarly as in \eqref{eq:nnnn1}. 

Recall that a sequence of action profiles is stable if players play the same action profile during each overlap $k=1, \dots, n$. Such a sequence can be identified with some $a^{[n]} \in A^n$. We define player $i$'s average discounted payoff from a stable sequence of action profiles by
$$
v_i^{\mathbf{M}}(a^{[n]}):=\frac{1}{\sum _{k=1}^{\overline{M}_n}\Delta ^{k-1}}
\left(\sum _{k=1}^n\Delta ^{\overline{M}_{k-1}} \left( \sum_{s=1}^{M_k}\Delta^{s-1}u_i(a^{i+k-1}) \right ) \right), \quad \forall a^{[n]} \in A^n. \label{eq:nnn2}
$$
Note that when $\mathbf{M} = (1, \dots, 1)$, $F(\delta, T, \mathbf{M})$ and $v_i^{M}$ reduce to $F(\delta, T)$ and $v_i$ in \Cref{sec:3}, respectively. We obtain the following theorems extending the previous one.

\begin{them}
	 For any $\delta \in (0,1]$, $T \in \mathbb{N}$ and $\mathbf{M}\in\mathbb{N}^n$, 
$$
F(\delta,T,\mathbf{M}) =co\left( \left \{v^{\mathbf{M}}(a^{[n]}): {a^{[n]}\in A^n} \right \} \right).
$$
\end{them}

That is, the feasible payoff set for the extended model is characterized similarly to the previous one, i.e., it is sufficient to consider length-$n$ sequences of action profiles regardless of $\delta$ and $T$, and now also regardless of overlapping structure $\mathbf{M}$. The idea behind this result is the same as in \autoref{them:1}, so we omit the proof.

\begin{them}
\label{them:4}
Fix any $\mathbf{M} \in \mathbb{N}^n$. For any $\delta, \delta' \in (0,1]$ and $T, T' \in \mathbb{N}$, if $\delta^T \leq  \delta{'}^{T'}$, $F(\delta', T',\mathbf{M}) \subseteq F(\delta, T,\mathbf{M})$. In particular, the following hold: 
\begin{enumerate}
	\item Given any $T \in \mathbb{N}$, $F(\delta ',T, \mathbf{M})\subseteq F(\delta ,T,\mathbf{M})$ for any $\delta, \delta'$ with $0< \delta <\delta ' \leq 1$.
\item Given any $\delta \in (0,1]$, $F(\delta ,T,\mathbf{M})\subseteq F(\delta ,T+1,\mathbf{M})$ for any $T \in \mathbb{N}$.	
\end{enumerate}
	
\end{them}
That is, the monotonicity of the feasible payoff set holds for general overlapping structures.

The basic idea of the proof is to regard the extended model as the previous model in \Cref{sec:2} with some ``dummies.'' For instance, consider an OLG repeated game $OLG (G, \delta, T=1, \mathbf{M}=(1,3))$, where $G$ is a stage game with two players $i=1,2$ for some payoff functions $u_1$ and $u_2$ (see \autoref{fig:aaaa1}). 
\begin{figure}[t]
\center
\scalebox{1.0}{
\begin{tabular}{|c|c|c|c|c|c|c|c|c|c|c|c|c|c|}\hline
Period&1&2&3&4&5&6&7&8&9&10&11&12&$\cdots$\\\hline
$A_1$&\multicolumn{4}{c|}{Generation 1}&\multicolumn{4}{c|}{Generation 2}&\multicolumn{4}{c|}{Generation 3}&$\cdots$\\\hline
$A_2$&0&\multicolumn{4}{c|}{Generation 1}&\multicolumn{4}{c|}{Generation 2}&\multicolumn{4}{c|}{Generation 3 $\ \cdots$}\\\hline
\end{tabular}
}
\caption{$OLG (G, \delta, T=1, \mathbf{M}=(1,3))$}
\label{fig:aaaa1}
\end{figure}
Now let $\tilde{G}$ be a stage game in which there are two dummy players, players 3 and 4, who have a singleton action set and whose payoffs are assumed to be constantly 0, in addition to players 1 and 2 (players 1 and 2's payoffs in $\tilde{G}$ are the same as in $G$). Then, consider an auxiliary OLG repeated game $OLG (\tilde{G}, \delta, T=1, \tilde{\mathbf{M}} = (1,1,1,1))$, where player $i \in \{1,2,3,4\}$ in generation 1 is born at $t=i$ (see \autoref{fig:aaaa2}). Note that this game has the OLG structure we studied in the previous sections. Since players 3 and 4 have payoffs of 0 always, their payoffs do not contribute to the calculation of the $\lambda$-weighted welfare in \eqref{eq:nn2}.

\begin{figure}[t]
\center
\scalebox{1.0}{
\begin{tabular}{|c|c|c|c|c|c|c|c|c|c|c|c|c|c|}\hline
Period&1&2&3&4&5&6&7&8&9&10&11&12&$\cdots$\\\hline
$A_1$&\multicolumn{4}{c|}{Generation 1}&\multicolumn{4}{c|}{Generation 2}&\multicolumn{4}{c|}{Generation 3}&$\cdots$\\\hline
$A_2$&0&\multicolumn{4}{c|}{Generation 1}&\multicolumn{4}{c|}{Generation 2}&\multicolumn{4}{c|}{Generation 3 $\ \cdots$}\\\hline
$\tilde{A}_3$&\multicolumn{2}{c|}{0}&\multicolumn{4}{c|}{Generation 1}&\multicolumn{4}{c|}{Generation 2}&\multicolumn{3}{c|}{Generation 3 $\ \cdots$}\\\hline
$\tilde{A}_4$&\multicolumn{3}{c|}{0}&\multicolumn{4}{c|}{Generation 1}&\multicolumn{4}{c|}{Generation 2}&\multicolumn{2}{c|}{Generation 3 $\ \cdots$}\\\hline
\end{tabular}
}
\caption{$OLG (\tilde{G}, \delta, T=1, \tilde{\mathbf{M}} = (1,1,1,1))$}
\label{fig:aaaa2}
\end{figure}

\section{Discussions}
\label{sec:7}

\subsection{Non-stationary Play}
\label{subsec:discuss}

Thus far, we have restricted our attention to stationary feasible payoffs. In this subsection, we discuss how the relaxation of this restriction would affect the monotonicity result in terms of $\delta$. We consider non-stationary sequences of players' action profiles, in which different generations of the same player may play different sequences of action profiles during their lifetime. One way to extend the concept of monotonicity in this case would be as follows: Given a sequence $(\bar{u}^y(\delta))_{y \in \mathbb{N}}$ of players' average discounted payoffs for each generation $y=1,2,\dots$, which is feasible with discount factor $\delta$, we ask whether the same sequence is feasible with $\delta' < \delta$. 

The following example shows that this is not the case: Consider the OLG repeated game with two players and two possible stage game payoffs of $(1,0)$ and $(0,1)$. Suppose $T=1$. Consider the following sequence of payoff vectors for each $t = 1,2, \dots$:
$(1,0), (1,0), (0,1), (0,1), \dots.$
 That is, the first two periods give $(1,0)$, followed by $(0,1)$ forever. The corresponding average payoff for each generation of player 1 is $\bar{u}_1^1 = 1$ for the first generation, and $\bar{u}_1^y= 0$ for any $y \geq 2$. Observe that player 2 in the first generation, who is born at $t=2$, obtains the average payoff $\bar{u}_2^1 = \frac{ \delta}{ 1+ \delta}$ and $\bar{u}_2^y = 1$ for any $y \geq 2$.

Now consider $\delta' < \delta$. Since $\bar{u}_1^1= 1$, the first two periods must give players $(1,0), (1,0)$. Note that the maximum payoff of player 2 in the first generation is obtained when $(0,1)$ is given at $t=3$, yielding the average payoff $\frac{\delta'}{1+\delta'}$, which is strictly smaller than $\frac{\delta}{ 1+ \delta}$.

 \subsection{Implications for Equilibrium Payoffs}
 \label{sec:6.2}

The monotonicity result of the feasible payoff set can also shed light on the equilibrium (Nash equilibrium or subgame perfect equilibrium) payoff set. It follows from our findings that the set of ``stationary'' equilibrium (i.e., each generation of a player employs the same strategy) payoffs is, in general, not increasing in $\delta$. Although it is obvious, we state this as a proposition below.\footnote{This contrasts with the case of repeated games with infinitely-lived players with a PRD \cite[]{APS_1990_ECMA}, where the subgame perfect equilibrium payoff set is monotonically increasing in players' discount factor. When there is no PRD, this monotonicity might not hold (see \cite{MOS_2002_GEB,Yamamoto_2010_IJGT}).}

\begin{prop}
Consider an OLG repeated game with the stage game $G$. Suppose that any extreme point of $V$ is a payoff vector of a static Nash equilibrium of $G$. Then the stationary subgame perfect equilibrium payoff set is non-increasing in $\delta$ and non-decreasing in $T$. 
\end{prop}

\begin{proof}
Under the hypothesis, the OLG feasible payoff set coincides with the stationary subgame perfect equilibrium payoff set for any given $\delta$ and $T$. Thus, \autoref{them:2} implies the result.
\end{proof}

Intuitively, if any extreme points of $V$ can be achieved by some Nash equilibrium, there is no issue of providing incentives to players to play a certain action. In general, to implement a certain action, appropriate incentives should be given by varying continuation payoffs. This suggests that there would be a trade-off of players being more patient in OLG repeated games: Being more patient, players are more easily disciplined as they care more about future payoffs. On the other hand, it is also costly, as it makes players more difficult to make intertemporal trades of their payoffs. This suggests that the ``optimal'' discount factor of players may be some intermediate value for some games.

Let us consider a simple example to illustrate the point made in the previous paragraph. Consider the OLG repeated games with $T=1$, where the stage game is described in \autoref{fig:6} and let $\lambda = (1, 1)$. We claim that the optimal $\delta$ that maximizes the $\lambda$-weighted sum of players' average discounted payoffs in (``stationary'' subgame perfect) equilibrium is $\delta = \frac{1}{2}$. To see this, we first observe that if $\delta \geq 1/2$, playing $(L, L),(R,R)$ (they are the action profiles that give the maximum sum of players' payoffs) for odd and even periods, respectively, can be achieved by a stationary subgame perfect equilibrium. This is because when players are old, they have no incentive to deviate; in addition, each player's action when they are young can be supported by punishing any deviation from it with $(M,M)$ (this is a one-shot Nash profile and gives the minmax payoff to each player): $2 + \delta  \geq 3 - \delta$ or $\delta \geq \frac{1}{2}$. As a result, the monotonicity result implies that the smallest $\delta = \frac{1}{2}$ yields the largest weighted sum of payoffs. 
\begin{figure}[]
\centering
\begin{tabular}{|c|c|c|c|}
\hline
    & $L$     & $R$     & $M$     \\ \hline
$L$ & $2,1$   & $-3,-3$ & $-3,-1$ \\ \hline
$R$ & $-3,-3$ & $1,2$   & $-5,3$  \\ \hline
$M$ & $3,-5$  & $-1,-3$ & $-1,-1$ \\ \hline
\end{tabular}
\caption{The stage game used in the example in \Cref{sec:6.2} }
\label{fig:6}
\end{figure}

\section{Conclusion} 
\label{sec:8}

In the present study, we analyze the feasible payoff set of OLG repeated games. First, we show that the set can be characterized by a convex combination of the average discounted payoffs of length-$n$ sequences of action profiles, where each of the action profiles is played for $T$ periods consecutively. Second, we find that the feasible payoff set is monotonically decreasing in $\delta$ and increasing in $T$.

A natural and important direction for future research is to study the set of equilibrium payoffs. In particular, it would be meaningful to extend our understanding of the trade-off of having more patient players, which is briefly described in \Cref{sec:6.2}. Relatedly, one might also think of a recursive characterization of the subgame perfect equilibrium payoff set, as in \cite{APS_1990_ECMA}. Making such a recursive characterization would be relatively straightforward. However, unlike the standard repeated games with infinitely-lived players, the object of the recursive characterization would be the set of the equilibrium payoff \emph{sequences} (of the length of $nT$) even for stationary equilibria, where each generation plays the same action, because players' ages keep changing. One might also consider the class of strongly symmetric equilibria \cite[]{APS_1986_JET} for symmetric stage games, as it allows a much simpler recursive characterization for standard repeated games. However, note that, in OLG repeated games, players are not symmetric, as their ages are different, even when the underlying stage game is symmetric.

\appendix
\section{Omitted Proofs}
\label{sec:app}

\subsection{Proof of \autoref{lem:1}}
\label{sec:proof_lem1}

\begin{proof}

Denote the ``age'' of player $i$ at overlap $k$ by $y_k (i) \in \{ 1,\dots, n\}$.

Consider overlap $k  \in \{ 1, \dots, n\}$ and $m \in \{ 1,\dots, n-1\}$. Let $k' \equiv k +m \text{ (mod $n$)}$. Then, from the optimality of $u^k$ at overlap $k$, 
		$$\sum_{ i =1}^n   \Delta^{ y_k (i)-1 } \lambda_i  u_i^k   \geq \sum_{i=1}^n \Delta^{y_k(i)-1} \lambda_i u_i^{k'} .$$
		By multiplying both sides by $\Delta^m$,
$$\sum_{ i =1}^n   \Delta^{ y_k (i)+ m-1 } \lambda_i  u_i^k  \geq \sum_{i=1}^n \Delta^{y_k(i) + m-1} \lambda_i u_i^{k'} .$$
Note that $y_{k'} (i) = y_k (i) + m$ if $y_k (i) + m \leq n$; otherwise $y_{k'} (i)= y_k (i) + m-n$.
\begin{multline*}
\sum_{ i : y_k(i) + m \leq n}   \Delta^{ y_{k'} (i)-1 } \lambda_i  u_i^k + \Delta^n \sum_{ i : y_k(i) + m >n}   \Delta^{ y_{k'} (i)-1 } \lambda_i  u_i^k \\
\geq \sum_{i : y_k(i) + m \leq n} \Delta^{y_{k'}(i)-1} \lambda_i u_i^{k'}  +\Delta^n\sum_{i : y_k(i) + m >n} \Delta^{y_{k'}(i)-1} \lambda_i u_i^{k'},	
\end{multline*}
or, equivalently, 
\begin{multline}
\label{eq:nnn6}
\sum_{ i : y_k(i) + m \leq n}   \Delta^{ y_{k'} (i)-1 } \lambda_i  u_i^k - \sum_{i : y_k(i) + m \leq n} \Delta^{y_{k'}(i)-1} \lambda_i u_i^{k'}  \\
\geq   \Delta^n \left ( \sum_{i : y_k(i) + m >n} \Delta^{y_{k'}(i)-1} \lambda_i u_i^{k'} -  \sum_{ i : y_k(i) + m >n}   \Delta^{ y_{k'} (i)-1 } \lambda_i  u_i^k \right).	
\end{multline}
On the other hand, from optimality at overlap $k'$, we have 
\begin{multline*}
\sum_{ i : y_k(i) + m \leq n}   \Delta^{ y_{k'} (i)-1 } \lambda_i  u_i^{k'}  + \sum_{ i : y_k(i) + m >n}   \Delta^{ y_{k'} (i)-1 } \lambda_i  u_i^{k'}  \\
\geq \sum_{i : y_k(i) + m \leq n} \Delta^{y_{k'}(i)-1} \lambda_i u_i^k  +\sum_{i : y_k(i) + m >n} \Delta^{y_{k'}(i)-1} \lambda_i u_i^k, 
\end{multline*}
or
\begin{multline}
\label{eq:nnn7}
\sum_{i : y_k(i) + m \leq n} \Delta^{y_{k'}(i)-1} \lambda_i u_i^k- \sum_{ i : y_k(i) + m \leq n}   \Delta^{ y_{k'} (i)-1 } \lambda_i  u_i^{k'}    \\
\leq     \sum_{ i : y_k(i) + m >n}   \Delta^{ y_{k'} (i)-1 } \lambda_i  u_i^{k'} -\sum_{i : y_k(i) + m >n} \Delta^{y_{k'}(i)-1} \lambda_i u_i^k.
\end{multline}
In order to satisfy both \eqref{eq:nnn6} and \eqref{eq:nnn7}, it must be 
$$\sum_{ i : y_k(i) + m >n}   \Delta^{ y_{k'} (i)-1 } \lambda_i  u_i^{k'}  \geq \sum_{i : y_k(i) + m >n} \Delta^{y_{k'}(i)-1} \lambda_i u_i^k,$$
or, equivalently,
$$\sum_{ i : y_{k'}(i) \in \{ 1, \dots, m\} }   \Delta^{ y_{k'} (i)-1 } \lambda_i  u_i^{k'}  \geq \sum_{i : y_{k'}(i) \in \{ 1, \dots, m\}} \Delta^{y_{k'}(i)-1} \lambda_i u_i^k .$$
Summing up both sides over $k \in \{1, \dots, n\}$, we have the inequality in the statement.
	\end{proof}

\subsection{Proof of \autoref{lem:2}}
\label{proof:lem2}
\begin{proof}

For $a^{[n]} \in \mathcal{A}^{*}_{\lambda}(\Delta)$, denote the numerator of $\frac{\partial W_{ \lambda}}{\partial \Delta } (a^{[n]},\Delta )$ by $\eta(\Delta)$. We will show that $\eta (\Delta) \leq 0$, thereby $\frac{\partial W_{ \lambda}}{\partial \Delta } (a^{[n]},\Delta ) \leq 0$. From \eqref{eq:nn3}, observe that for $u^{[n]} \in \mathcal{U}^{*}_{\lambda}(\Delta)$,
\begin{align*}
\eta (\Delta) &= \left (\sum_{k=1}^n \frac{d\Delta^{k-1} }{d\Delta }w_k (u^{[n]}) \right) \left(\sum_{m=1}^n \Delta^{m- 1} \right) - \left( \frac{d \sum_{m=1}^{n} \Delta^{m-1}} { d\Delta}\right)  \sum_{k=1}^n \Delta^{k-1} w_k(u^{[n]})\\
&=\left (\sum_{k=1}^n (k-1)\Delta^{k-2} w_k(u^{[n]}) \right) \left(\sum_{m=1}^n \Delta^{m-1} \right) - \left(\sum_{k=m}^{n} (m-1) \Delta^{m-2}\right)  \sum_{k=1}^n \Delta^{k-1} w_k(u^{[n]})\\
&=\sum_{k=1}^n \left ((k-1) \Delta^{k-2} \sum_{m=1}^{n} \Delta^{m-1} - \left( \sum_{m=1}^n (m-1) \Delta^{m-1}  \right) \Delta^{k-2}  \right) w_k(u^{[n]})\\
&=\sum_{k=1}^n  \left( \sum_{m=1}^n (k-m) \Delta^{k+m-3} \right)w_k(u^{[n]}).
\end{align*}
Recall the $m$-th constraint, $m \in \{ 1, \dots, n-1\}$, is
$$-\sum_{k=1}^m \Delta^{k-1} w_k (u^{[n]})  + \sum_{k=1}^m \Delta^{k-1} w_{n-m+k} (u^{[n]})\leq 0.   \quad (\# m)$$
We will show that 
$$\eta (\Delta) = \sum _{m=1}^{n-1}p_m\times (LHS\ of\ \# m),$$
where for $1\leq m\leq n-1$
$$p_m \equiv \Delta ^{n-2}+\Delta ^{n-3}+\cdots +\Delta ^{n-m-1}.$$
Note that as $p_m \geq 0$, this implies $\eta (\Delta) \leq 0$. To see the equality, fix $j \in \{ 1,\dots, n\}$. Observe that for any constraint $m\geq j$, there exists a non-positive term $- \Delta^{j-1}w_j (u^{[n]})$. In addition, for any constraint $m$ with $j \geq n-m+1$, there exists a non-negative term $\Delta^{j-n+m-1} w_j (u^{[n]})$. Therefore, the coefficient of $w_j (u^{[n]})$ is:
\begin{align*}
&\sum_{m = j}^{n-1} p_m (-\Delta ^{j-1} ) + \sum_{m =n+1-j}^{n-1} p_m \Delta^{j-n+m-1}\\
&=\sum_{m = j}^{n-1} (\Delta^{n-2} + \cdots + \Delta^{n-m-1}) (-\Delta^{j-1}) + \sum_{ m=n+1-j}^{n-1} (\Delta^{n-2} + \cdots + \Delta^{n-m-1}) \Delta^{j-n+m-1} \\
&=\sum_{ m=j}^{n-1} (-\Delta^{n+j-3} - \cdots - \Delta^{n-m+j-2})  + \sum_{m=n+1-j}^{n-1} (\Delta^{j+m-3} +\cdots + \Delta^{j-2})\\
&=(- \Delta^{n+j-3}- \cdots - \Delta^{n-2}) + (- \Delta^{n+j-3} - \cdots - \Delta^{n-1}) + \cdots + (-\Delta^{n+j-3} - \cdots - \Delta^{j-1})  \\
&+ (\Delta^{n-2} + \cdots + \Delta^{j-2}) + (\Delta^{n-1} + \cdots + \Delta^{j-2})  + \cdots + (\Delta^{j+n-4} + \cdots + \Delta^{j-2})\\
&=- \Delta^{n+j-3} ((n-j)-0) - \Delta^{n+j-4} ((n-j)-1)- \cdots - \Delta^{n-2} ((n-j) -(j-1))\\
&- \Delta^{n-3} ((n-j-1)- (j-1))- \cdots - \Delta^{j-1} (1- (j-1))- \Delta^{j-2} (0-(j-1))\\
&=\sum_{m=1}^n (j-m) \Delta^{j+m-3}.
\end{align*}
\end{proof}

\subsection{Proof of \autoref{prop:nnn1}}
\label{proof:prop1}

Let $S^0 \equiv \{a^{[n]}\in A^n: u(a^1)=u(a^2)=\dots =u(a^n)\}\subset A^n$ and $S^C \equiv A^n\setminus S^0$. First, we provide a few lemmas.

\begin{lem}
\label{lem:nn3}
$V \neq  F^*$ if and only if there exist two points $u$, $\tilde{u} \in V$ with $u \neq \tilde{u}$ that satisfy the following two conditions:
\begin{enumerate}
	\item $u$ and $\tilde{u}$ are vertices of an edge of $V$, i.e., there exists $\lambda \in \mathbb{R}^n \setminus \{ \mathbf{0} \}$ such that $\bar{u} \in \argmax_{ u' \in V} \lambda  \cdot u'$ holds if and only if
$\bar{u}  \in \{  t u +(1- t)\tilde{u} : t \in [0,1]\}.$
\item There exist at least two players $i$ such that $u_i \neq \tilde{u}_i$.
\end{enumerate}
	
\end{lem}

\begin{proof}
	$(\Leftarrow)$ Suppose not, i.e., $V = F^*$. Observe that any $u$ and $\tilde{u}$ on an edge of $F^*$ can have only one player whose payoff differs.
	
	$(\Rightarrow)$ Suppose not. Then, either $V$ does not have an edge or any edge of $V$ does not satisfy Condition 2. In the first case, trivially $V = F^*$. In the second case, there is only one player $i$ such that $u_i \neq \tilde{u}_i$. That is, any edge is parallel to some axis. Since $V$ is convex, this implies that $V$ is an $n$-dimensional cube. In turn, this implies that $V = F^*$, as $V$ contains the best and worst stage payoff profile for each player $i$ (i.e., $u(a)$ for some $a \in \argmax_{a' \in A} u_i (a')$ and $u(a)$ for some $a \in \argmin_{a' \in A} u_i (a')$). A contradiction. 
\end{proof}

\begin{lem}
\label{lem:nn4}
Suppose that there exist $\lambda \in\mathbb{R}^n \setminus \{\mathbf{0}\}$ and two points $u$, $\tilde{u} \in V$ that satisfy Conditions 1 and 2 in \autoref{lem:nn3}. Then, for any $\Delta \in (0,1)$, $W^*_\lambda (\Delta )>\lambda \cdot u$.
\end{lem}
\begin{proof}
Without loss of generality, let $u=(0,\dots,0)$. Note that $\lambda \cdot u = \lambda \cdot \tilde{u} =0$. Let $a, \tilde{a} \in A$ be such that $u(a) = u$ and $u(\tilde{a} ) = \tilde{u}$. If $\tilde{a}$ (resp. $a$) is played when player $n$ is the youngest (resp. not the youngest), players' average payoff vector is $u' \equiv c(\Delta ^{n-1}\tilde{u}_1,\Delta ^{n-2} \tilde{u}_2,\dots,\Delta \tilde{u}_{n-1},\tilde{u}_n)$, where $c \equiv \frac{1}{1+\Delta+\dots +\Delta ^{n-1}}$. Instead, if $a$ (resp. $\tilde{a}$) is played when player $n$ is the youngest (resp. not the youngest), the average payoff vector is $u'' \equiv \tilde{u}-c(\Delta ^{n-1}\tilde{u}_1,\Delta ^{n-2}\tilde{u}_2,\dots, \Delta \tilde{u}_{n-1}, \tilde{u}_n)= \tilde{u}-u'$.


Observe that for any $\Delta \in (0,1)$, $u'$ is not on the line segment between $u$ and $\tilde{u}$. This is because $\tilde{u}$ has at least two non-zero components due to Condition 2, and these components are multiplied by $\Delta$ different numbers of times. Thus, $\lambda \cdot u' \neq 0$. Since $\lambda \cdot u'' = \lambda \cdot \tilde{u} - \lambda \cdot u'=- \lambda \cdot u'$, either $\lambda \cdot u' >0$ or $\lambda \cdot u'' >0$. Therefore, $W^*_\lambda (\Delta) \geq \max\{ \lambda \cdot u', \lambda \cdot u'' \}>0.$
\end{proof}

\begin{lem}
\label{lem:nnn5}
When $\mathcal{A}^*_\lambda (\Delta )\subset S^C$ for some $\lambda \in \mathbb{R}^n \setminus \{ \mathbf{0} \}$, the following strict inequality holds for some $m\in \{1,\dots,n-1\}$:\\
$$-\sum_{k=1}^m \Delta^{k-1} w_k (u^{[n]})  + \sum_{k=1}^m \Delta^{k-1} w_{n-m+k} (u^{[n]})<0.$$
\end{lem}
\begin{proof}
We use proof by contradiction. Suppose that the above inequality does not hold for any $m$, i.e., for any $m\in \{1,\dots,n-1\}$,
\begin{equation}
\label{eq:nn11}
\sum_{k=1}^m \Delta^{k-1} w_k (u^{[n]})  - \sum_{k=1}^m \Delta^{k-1} w_{n-m+k} (u^{[n]})=0.
\end{equation}
We claim that the above equalities imply $w_1(u^{[n]})=w_2(u^{[n]})=\cdots =w_n(u^{[n]})$. To show this, let us regard $w_n (u^{[n]})$ as a parameter and the rest $w_1( u^{[n]}), \dots, w_{n-1} (u^{[n]} )$ as unknowns. Clearly, $w_1 (u^{[n]}) = w_n (u^{[n]}), \cdots, w_{n-1} (u^{[n]}) = w_n (u^{[n]})$ is a solution of the system. Next, we show that it is a unique solution. Consider the submatrix consisting of columns 1 to $n-1$ of the coefficient matrix (i.e., the $(n-1)\times n$ matrix in which the $m$-th row corresponds to the coefficients in equation \eqref{eq:nn11} for $m$), which is an $(n-1) \times (n-1)$ matrix. Our purpose is to show that this has the full rank. Let $r_k$ be the $k$-th row of the submatrix, $k=1, \dots, n-1$. To show that this matrix is invertible, we show that $r_1, \dots, r_{n-1}$ are linearly independent or, equivalently, $(\alpha_k)_k \in \mathbb{R}^{n-1}$ satisfies $\sum_{k=1}^{n-1} \alpha_k r_k = \mathbf{0}$ if and only if $\alpha_k = 0$ for all $k$. 
By rearranging each column $l = 1, \dots, n-1$ of $\sum_{k=1}^{n-1} \alpha_k r_k =\mathbf{0}$, we have:
\begin{align*}
(\alpha_1 + \cdots + \alpha_{n-1})  &=0,\\
\Delta (\alpha_2 + \cdots + \alpha_{n-1}) &=\alpha_{n-1},\\
\Delta^2 (\alpha_3 + \cdots + \alpha_{n-1}) &= \alpha_{n-2} + \Delta \alpha_{n-1},\\
\vdots\\
\Delta^{n-2} \alpha_{n-1} &= \alpha_2 + \cdots + \Delta^{n-3}  \alpha_{n-1}.
\end{align*}
Using $\alpha_1 + \alpha_2 + \cdots + \alpha_{n-1} =0$,
\begin{align*}
- \alpha_1 \Delta &= \alpha_{n-1},\\
- (\alpha_1 + \alpha_2) \Delta^2 &= \alpha_{n-2} + \alpha_{n-1} \Delta,\\
\vdots\\
- (\alpha_1 +  \cdots + \alpha_{n-2})\Delta^{n-2} &= \alpha_2 + \alpha_3 \Delta   + \cdots +  \alpha_{n-1}\Delta^{n-3}.
\end{align*}
By substituting each equation into the next, we have:
\begin{align*}
- \alpha_1 \Delta &= \alpha_{n-1},\\
- \alpha_2 \Delta^2 &= \alpha_{n-2},\\
\vdots\\
- \alpha_{n-2} \Delta^{n-2} &= \alpha_2.	
\end{align*}
Then, by comparing the equations with the same variables (for instance, the second with the last), we obtain $\alpha_2 =\alpha_3 = \cdots = \alpha_{n-2}=0$, and using the first equation and $\alpha_1 + \cdots + \alpha_{n-1}=0$, we have $\alpha_1 = \alpha_{n-1}=0$. This concludes the proof of the claim.

By the definition of each $w_k(u^{[n]})$ and the above claim, playing $(u^1,\cdots, u^1)$ yields the same $\lambda$-weighted score as playing $u^{[n]}$. Therefore, $(u^1,\cdots, u^1)\in\mathcal{U}^* _{\lambda}(\Delta )$ must hold, and there exists an action profile $a\in A$ that satisfies $u^1=u(a)$ and $(a,\cdots ,a)\in\mathcal{A}^*_\lambda (\Delta )$. This contradicts $\mathcal{A}^*_\lambda (\Delta )\subset S^C$, completing the proof.
\end{proof}

\begin{proof}[Proof of \autoref{prop:nnn1}] 
Observe that \autoref{lem:nn3} and \autoref{lem:nn4} imply that $V\neq F^*$ if and only if for any $\delta  \in (0,1)$ and $T\in \mathbb{N}$, $F(\delta,T )\neq V$.

Suppose $\delta  \in (0, 1)$, $T\in \mathbb{N}$, and $F^* \neq V$ so that $F(\delta,T )\neq V$. Then, clearly for some $\lambda \neq 0$, $\mathcal{A}^*_\lambda (\Delta )\subset S^C$, where $\Delta=\delta^T$.
Then, by \autoref{lem:nnn5}, the inequality \eqref{eq:2} strictly holds for at least one particular direction $\lambda \in \mathbb{R}^n \setminus \{ \mathbf{0} \}$, and this concludes the proof.
\end{proof}

\subsection{Proof of \autoref{them:4}}

\begin{proof}
Consider an OLG repeated game $OLG (G, \delta, T, \mathbf{M})$. Consider an auxiliary OLG repeated game defined as follows. For each $i \in N$, we divide $i$'th overlap (whose length is $M_i T$) into smaller ``dummy overlaps'' $(i,l_i)$, $l_i = 1, \dots, M_i$, with the same length of $T$. Then, there are $\sum_{i=1}^n M_i$ dummy overlaps. In addition, for some dummy overlap $(i,l_i)$, if no player is born, then we introduce a ``dummy player'' $(i,l_i)$ who is born at the beginning of the dummy overlap $(i,l_i)$. Assume that this dummy player has a single available action, and the payoff function assigns $0$ to any action profile. Let $\tilde{G}$ be the stage game with such dummy players, where the original players' available actions and payoffs are the same as those of $G$. Since in every dummy overlap, either an original or a dummy player is born, this OLG repeated game with the dummy overlaps and players can be regarded as the OLG repeated game $OLG (\tilde{G}, \delta, T)$ in \Cref{sec:2}. Therefore, we can apply our previous results to this auxiliary game. Thus, for each direction $\lambda$, the $\lambda$-weighted welfare is monotonically increasing as $\delta$ decreases and $T$ increases by \autoref{them:2}. Lastly, we observe that the welfare is the same as the one without the dummy players, as their payoffs are assumed to be 0. 
\end{proof}

\bibliography{FeasibleOLG_bib}
\bibliographystyle{chicago}

\end{document}